\documentclass[conference]{IEEEtran}

\usepackage{etex}

\usepackage{amsfonts}
\usepackage{times}
\usepackage{latexsym}
\usepackage{amssymb}
\usepackage{amsmath}
\usepackage{cite}
\usepackage{verbatim}
\def\bb0{{\mathbb{0}}}

\def\ba{{\mathbf{a}}}
\def\bb{{\mathbf{b}}}

\def\b0{{\mathbf{0}}}

\def\bS{{\mathbf{S}}}

\def\b1{{\mathbf{1}}}

\def\bbE{{\mathbb{E}}}

\def\sf0{{\mathsf{0}}}

\usepackage{amssymb}
\usepackage{amsxtra}
\usepackage{amsmath}
\usepackage{amsthm}
\usepackage{mathtools}

\newtheorem{theorem}{Theorem}

\newtheorem{Lemma}[theorem]{Lemma}

\newtheorem{definition}[theorem]{Definition}

\usepackage{amsmath,amssymb}
\usepackage{graphicx}
\usepackage{color}
\usepackage{cite}
\usepackage{multirow,tabularx}
\usepackage{ifthen}
\usepackage{times}
\usepackage{graphicx}
\usepackage{amsbsy} 
\usepackage{epsfig}
\usepackage{color}

\usepackage{stackrel}

\usepackage{amsfonts}
\usepackage{stfloats}
\usepackage[all]{xy}

\usepackage{tikz}
\usetikzlibrary{arrows,snakes}
\usepackage{pgfplots}
\pgfplotsset{compat=1.4}

\usetikzlibrary{spy}

\newcommand{\M}{\mathcal{M}}
\newcommand{\A}{\mathcal{A}}

\newcommand{\off}{\text{off}}

\newcommand{\worst}{\text{worst}}

\begin{document}
\title{Sub-Modularity of Waterfilling with Applications to Online
  Basestation Allocation}
\author{\IEEEauthorblockN{Kiran Thekumparampil and Andrew Thangaraj}
\IEEEauthorblockA{Department of Electrical Engineering\\
Indian Institute of Technology Madras, Chennai, India\\
Email: ee10b121,andrew@ee.iitm.ac.in}
\and
\IEEEauthorblockN{Rahul Vaze}
\IEEEauthorblockA{School of Technology and Computer Science\\
Tata Institute of Fundamental Research, Mumbai, India\\
Email: vaze@tcs.tifr.res.in}}
\maketitle
\begin{abstract} We show that the popular water-filling
  algorithm for maximizing the mutual information in parallel Gaussian
  channels is sub-modular.  The sub-modularity of water-filling
  algorithm is then used to derive \emph{online} basestation
  allocation algorithms, where mobile users are assigned to
  one of many possible basestations immediately and irrevocably 
  upon arrival without knowing the future user information. The goal of
  the allocation is to maximize the sum-rate of the system under power
  allocation at each basestation. We present online algorithms with competitive
  ratio of at most 2 when compared to offline algorithms that have knowledge of all future user arrivals. 
\end{abstract}
\section{Introduction}
\label{sec:introduction}
In  combinatorial optimization, sub-modular functions play the role of convex functions in continuous optimization. For a
sub-modular function, the incremental gain from adding an extra
element in the set decreases with the size of the set. The interest in sub-modular functions is because results in combinatorial optimization show that greedy algorithms are close to optimal algorithms with provable guarantees \cite{Nemhauser1978, Lehmann2006combinatorial}.

In this paper, we show that the water-filling algorithm for maximizing
the mutual information in parallel Gaussian channels is
sub-modular. Log-based utility functions arising from the
capacity of a Gaussian channel are used in resource
allocation for which water-filling is the optimal solution. Thus, the
sub-modularity of water-filling has
widespread applications in combinatorial resource allocation. We present one example on online basestation allocation of
mobile users to basestations with $\log$-utility based power allocation.

Specifically, we consider the {\it online} downlink basestation association
problem, where each user on its arrival reveals its SNRs to each of the basestations,
and is allocated to one of the basestations for maximizing the
sum-rate at the end of all user arrivals. Each user is allocated
immediately upon arrival, and the association once made cannot be revoked. 

An \emph{online} algorithm allocates users causally without information about
future user arrivals. On the other hand, in an offline algorithm, all
future user arrivals and rates are revealed in the beginning. The performance of an online algorithm is
characterized by its \emph{competitive ratio}, which is the
ratio of the utility of the offline algorithm to that of the online algorithm
\cite{BorodinOnlineBook}. The problem is then to find online
algorithms with smallest possible competitive ratio. 

Associating mobile users to basestations under different utility
models is a classical problem in the
literature. Many utility models have been considered in
prior work including load balancing \cite{DeVeciana2009LoadBalancing,
  Altman2011LoadBalancing, DeVeciana2012userAssociation}, cell
breathing \cite{Bejerano2009CellBreathing}, call admission
\cite{Wu2000MACA} and fairness \cite{Bejerano2004Fairness}, sub-carrier/power allocation either jointly
with base-station allocation \cite{Subramanian2009userAssociation} or without it \cite{Kim2005SpectrumPowerAllocation}\cite{Yates2009SpectrumAllocation}.
Most of the prior work on allocation 
assumes either exact user information or statistics is known
and formulates a joint optimization problem. In contrast, design and
analysis of online algorithms for the basestation association problem do not require any
information or assumption about the statistics of the user's profile.

Online algorithms have been designed for many related problems in literature, e.g.,  load balancing \cite{Azar1998LoadBalancing}, load balancing with deadlines \cite{Moharir2013},  
maximum weight matching \cite{Korula2009}, picking best subset of fixed cardinality \cite{babaioff2007k-Secretary} (called the $k$-secretary problem), and 
multi-partitioning \cite{Nemhauser1978, Lehmann2006combinatorial}. Our contributions are as follows:
\begin{itemize}
\item We prove that the water-filling function that corresponds to
  maximizing the mutual information of parallel Gaussian channel under
  a sum-power constraint is sub-modular. To the best of our knowledge, the sub modularity of the waterfilling algorithm is not known in literature, and is an important result with several ramifications in combinatorial resource allocation.

\item We exploit the sub-modularity of the water-filling function to derive a $2$-competitive online basestation allocation algorithm for any input (possibly chosen by an adversary)
by using results from \cite{Nemhauser1978, Lehmann2006combinatorial}.
\end{itemize}

\section{Sub-Modularity of the Waterfilling Function}
Consider $M$ parallel Gaussian channels 
$$Y_i = X_i + Z_i,\ \ \ i=1,2,\ldots,M,$$ 
with noise variance $\bbE\{Z_i^2\} = N_i > 0$, and sum-power constraint $\bbE\{\sum_{i=1}^M X_i^2\} \le P$. 
To maximize the mutual information between the input and the output
over a subset of channels $S\subseteq \{1,2,\ldots,M\}$, we have to solve 
\begin{align}
  \label{eq:7}
\max\;\; &R(S)=\sum_{i\in S}\log\left(1+\frac{P_i}{N_i}\right)\\
\nonumber \text{subject to }&\sum_{i\in S} P_i \le P, \ P_i \geq 0,\ i\in S. 
\end{align}
The optimal solution to \eqref{eq:7} is given by
$P^*_i=(\nu(S)-N_i)^+$, and $\nu(S)$ is the so-called water level $\nu$ chosen to satisfy
\begin{equation}
\sum_{i\in S} (v(S) - N_i)^+ = P,
\label{eq:14}
\end{equation}
with $(x)^+=x$, for $x>0$ and $0$ otherwise, and the optimal objective function is
\begin{equation}
  \label{eq:8}
  R^*(S)=\sum_{i\in S} \log\left(1+ \frac{(\nu(S)-N_i)^+}{N_i}\right).
\end{equation}
The optimal power allocation $P^*$ is popularly known as the waterfilling, and we call $R^*(S)$ as the waterfilling function. Note that $R^*(S)$ is a set function from the power set of $\{1,2,\ldots,M\}$ to the real numbers. Also note that,
\begin{equation}
  \label{eq:9}
  R^*(S)=\sum_{i\in S: \nu(S)>N_i} \log\left(\frac{\nu(S)}{N_i}\right)=\log\frac{\nu(S)^{|T(S)|}}{\prod_{i\in T(S)}N_i},  
\end{equation}
where 
\begin{equation}
  \label{eq:5}  
T(S)=\{i\in S: \nu(S)>N_i\},
\end{equation}
denotes the channels in $S$ that are allotted non-zero power. Using the definition of $T(S)$ in (\ref{eq:14}), we get that 
\begin{equation}
  \label{eq:13}
  \sum_{i\in T(S)} (v(S) - N_i) = |T(S)|\nu(S)-\sum_{i\in T(S)}N_i = P.
\end{equation}
\begin{definition}\label{definition:submod}
Let $U$ be a finite set, and let $2^U$ be the power set of $U$. A real-valued set function $f:2^U\to\mathbb{R}$ is said to be {\it monotone} if $f(S) \le f(T)$ for $S \subseteq T \subseteq U$,  and {\it submodular} if 
\begin{equation}
  \label{eq:22}
f(S)+f(T)\ge f(S\cap T)+f(S\cup T),\ \forall S,T\in 2^U.   
\end{equation}
An equivalent definition of sub-modularity is 
\begin{equation}\label{eq:usefuldefnsubmod}
f(S \cup \{i\}) + f(S \cup \{j\}) \ge f(S) + f(S \cup \{i, j\})
\end{equation}
for every $S\subseteq U$ and every $i,j\in U\setminus S$ with $i\ne j$. 
\end{definition}

\begin{theorem}\label{thm:submod} Water-filling function $R^*(\cdot)$ is sub-modular.
\end{theorem}  
\begin{proof} We will show the sub-modularity of $R^*(\cdot)$ by showing that $R^*(\cdot)$ satisfies \eqref{eq:usefuldefnsubmod}. There are four different waterfilling function evaluations in (\ref{eq:usefuldefnsubmod}),  that involve the following four water levels 
 $\nu_i=\nu(S\cup\{i\})$, $\nu_j=\nu(S\cup\{j\})$, $\nu=\nu(S)$, and $\nu_{ij}=\nu(S\cup\{i,j\})$, respectively.
\begin{figure}[ht!]
\begin{center}
\begin{tikzpicture}[every node/.style={scale=1}]
\begin{scope}[node distance=2cm,>=angle 90]
\def\l{0.7}
\draw[thick,draw=brown] (0,0) --node[below]{$j$} (\l,0); 
\def\yi{3*\l+\l/2}
\draw[very thick,draw=brown] (0,0) rectangle node{$N_j$} (\l,\yi); 
\foreach \i in {1,...,6}
{
  \def\x{\l*\i}
  \draw[very thick,draw=brown] (\x,0) --node[below]{\i} (\x+\l,0); 
  \draw[very thick,draw=brown] (\x,0) rectangle node{$N_{\i}$} (\x+\l,\l+\x);   
}
\def\x{7*\l}
\draw[very thick,draw=brown] (\x,0) --node[below]{$i$} (\x+\l,0); 
\draw[thick,draw=brown] (\x,0) rectangle node{$N_i$} (\x+\l,\l); 
\def\y{8*\l}
\draw[very thick,draw=brown] (0,0) -- (0,\y);
\draw[very thick,draw=brown] (\x+\l,0) -- (\x+\l,\y);
\path[fill=gray] (0,\yi) -- ++(\l,0) -- (\l,2*\l) -- ++(\l,0) -- ++(0,\l) -- ++(\l,0) -- ++(0,\l) -- ++(\l,0) -- ++(0,\l/2) -- ++(-4*\l,0) -- (0,\yi);
\def\vij{4*\l+\l/2}
\path[fill=gray] (\x,\l) rectangle (\x+\l,\vij);
\def\vj{\vij+\l}
\path[fill=gray!75!white] (\l,\vij) -- ++(3*\l,0) -- ++(0,\l/2) -- ++(\l,0) -- ++(0,\l/2) -- ++(-4*\l,0) -- (\l,\vij);
\path[fill=gray!75!white] (\x,\vij) rectangle (\x+\l,\vj);
\def\vi{\vj+\l}
\path[fill=gray!50!white] (0,\vij) -- ++(\l,0) -- ++(0,\l) -- ++(4*\l,0) -- ++(0,\l/2) -- ++(\l,0) -- ++(0,\l/2) -- ++(-6*\l,0) -- (0,\vij);
\path[fill=gray!25!white] (\l,\vi) -- ++(5*\l,0) -- ++(0,\l/2) -- ++(\l,0) -- ++(0,\l/2) -- ++(-6*\l,0) -- (\l,\vi);
\draw[very thick,dashed,draw=brown] (\l,\yi) -- (\l,\y);
\draw[very thick,dashed,draw=brown] (\x,\y-\l) -- (\x,\y);
\draw[draw=blue] (0,\vij) -- node[above]{$\nu_{ij}$} (4*\l,\vij);
\draw[draw=blue] (\x,\vij) -- node[above]{$\nu_{ij}$} (\x+\l,\vij);
\draw[draw=blue] (\l,\vj) -- node[above]{$\nu_{i}$} (5*\l,\vj);
\draw[draw=blue] (\x,\vj) -- node[above]{$\nu_{i}$} (\x+\l,\vj);
\draw[draw=blue] (0,\vi) -- node[above]{$\nu_{j}$} (6*\l,\vi);
\draw[draw=blue] (\l,\vi+\l) -- node[above]{$\nu$} (7*\l,\vi+\l);
\end{scope}
\end{tikzpicture}
\end{center}
 \caption{Illustration of water levels, $S=\{1,2,3,4,5,6\}$.}
 \label{fig:waterlevels}
 \end{figure}
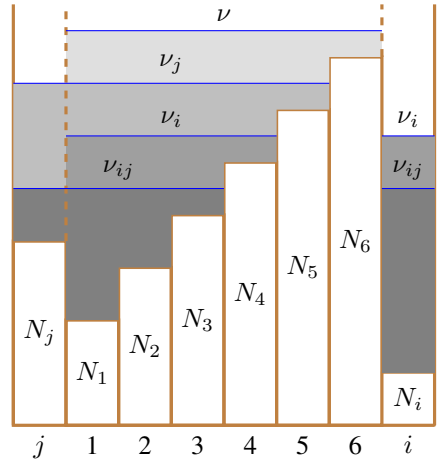
As illustrated  in Fig. \ref{fig:waterlevels}, we have the following relationships among the water levels:
\begin{equation}
  \label{eq:11}
  \nu_{ij}\le \nu_i\le \nu,\ \nu_{ij}\le\nu_j\le \nu.
\end{equation}
Without loss of generality, we assume that $\nu_j\ge\nu_i$. So, we have
\begin{equation}
  \label{eq:3}
  \nu_{ij}\le \nu_i\le \nu_j\le \nu.
\end{equation}
Similarily, let the channels receiving non-zero power be $T=T(S)$,
$T_i=T(S\cup\{i\})$, $T_j=T(S\cup\{j\})$ and
$T_{ij}=T(S\cup\{i,j\})$. Finally, let the optimal rates be $R^*=R^*(S)$, $R^*_i=R^*(S\cup\{i\})$, $R^*_j=R^*(S\cup\{j\})$ and $R^*_{ij}=R^*(S\cup\{i,j\})$. 

Also, the optimal rates satisfy the following relationship:
\begin{equation}
  \label{eq:12}
  R^*_{ij}\ge R^*_i\ge R^*,\ R^*_{ij}\ge R^*_j\ge R^*.
\end{equation}
Using the above relationships, we first settle two easy cases. If $i\notin T_{ij}$, then $R^*_{ij}=R^*_j$, and the submodularity condition (\ref{eq:usefuldefnsubmod}) is seen to be readily satisfied using (\ref{eq:12}). By a similar argument, if $j\notin T_{ij}$, then $R^*_{ij}=R^*_i$ and hence (\ref{eq:usefuldefnsubmod}) is clearly satisfied. 

So, in what follows, we will assume that $i,j\in T_{ij}$. If $i,j\in T_{ij}$, we have $\nu_{ij}>N_i$ and $\nu_{ij}>N_j$, which, by \eqref{eq:11}, results in $\nu_i>N_i$ and $\nu_j>N_j$. This implies that $i\in T_i$ and $j\in T_j$. Now, let $\bar{T}_{ij}=T_{ij}\setminus\{i,j\}$, $\bar{T}_i=T_i\setminus\{i\}$ and $\bar{T}_j=T_j\setminus\{j\}$. From (\ref{eq:11}), it is clear that $\bar{T}_{ij}\subseteq \bar{T}_i$ and $\bar{T}_j\subseteq T$. Further, letting $\hat{T}_i=\bar{T}_i\setminus\bar{T}_{ij}$ and $\hat{T}=T\setminus\bar{T}_j$, we can write
\begin{align}
  \label{eq:1}
  &T_i=\{i\}\sqcup \bar{T}_{ij}\sqcup \hat{T}_i,\ T_j=\{j\}\sqcup \bar{T}_j,\\
\label{eq:2}  &T=\bar{T}_j\sqcup \hat{T},\ \ \ \ \ T_{ij}=\{i,j\}\sqcup \bar{T}_{ij},
\end{align}
where $\sqcup$ denotes disjoint union. 

After a series of simplifications shown in Appendix \ref{sec:simpl-cond-subm}, we reduce the condition for submodularity of $R^*$ to the following:
\begin{equation}
  \nu_i^{|T_i|}\ \nu_j^{|T_j|}\ \prod_{l\in \hat{T}}N_l\ge \nu^{|T|}\ \nu_{ij}^{|T_{ij}|}\ \prod_{m\in \hat{T}_i}N_m.
\label{eq:17}
\end{equation}
The terms involved in \eqref{eq:17} can be ordered in a specific way, and this is stated in the next Lemma:
\begin{Lemma}\label{lem:ord}
The following ordering holds true for every $m\in\hat{T}_i$ and every $l\in\hat{T}$:
\begin{equation}
  \label{eq:4}
  \nu_{ij}\overset{(a)}{\le} N_m\overset{(b)}{\le} \nu_i \overset{(c)}{\le} \nu_j\overset{(d)}{\le} N_l\overset{(e)}{\le} \nu.  
\end{equation}
\end{Lemma}
\begin{proof}
See Appendix \ref{sec:proof-lemma-ord}.
\end{proof}
In particular, by Lemma \ref{lem:ord}, all terms in the LHS of
\eqref{eq:17} are bounded between two terms $N_{m^*}$ and $\nu$ in the
RHS, where $m^*=\arg\max_{m\in\hat{T}_i}N_m$, is the highest
noise level in $\hat{T}_i$. This is a crucial observation, and is
exploited later. The next Lemma states two equalities on the terms involved in \eqref{eq:17}.
\begin{Lemma}\label{lem:sum}
The following relationships hold true:
\begin{align}
    \label{eq:18}
    |T_i|\nu_i+|T_j|\nu_j+\sum_{l\in\hat{T}}N_l &= |T|\nu+|T_{ij}|\nu_{ij}+\sum_{m\in\hat{T}_i}N_m,\\
\label{eq:26} |T_i|+|T_j|+|\hat{T}|& =|T|+|T_{ij}|+|\hat{T}_i|.
\end{align}
\end{Lemma}
\begin{proof}
See Appendix \ref{sec:proof-lemma-sum}.
\end{proof}
In words, Lemma \ref{lem:sum} states that the number of terms and their sum in the LHS and RHS of \eqref{eq:17}, counting multiplicities, are equal. The next Lemma is a general version of the inequality in \eqref{eq:17} with constraints motivated by Lemmas \ref{lem:ord} and \ref{lem:sum}.
\begin{Lemma}
Let $a_i$, $b_i$, $1\le i\le n$, be positive real numbers with $a_1 \leq a_2 \leq \ldots\leq a_n$, $b_1\leq b_2 \leq\ldots\leq b_n$,  and
\begin{equation}
\sum^n_{i=1} a_i = \sum^n_{j=1} b_j. 
\label{eq:28}
\end{equation}
If $\exists$ $k$, $1 \leq k\leq n-1$, for which 
\begin{equation}
a_k\leq b_1 \leq b_2 \leq \ldots \leq b_n \leq a_{k+1}, 
\label{eq:inside}
\end{equation}
then, we have,
\begin{equation}
\prod^n_{j=1} b_j \geq \prod^n_{i=1} a_i. 
\label{eq:prod}
\end{equation}
\label{lem:maj}
\end{Lemma}
\begin{proof}
See Appendix \ref{sec:proof-lemma-maj}.
\end{proof}
To show the sub-modularity of $R^*$, i.e. to prove (\ref{eq:17}), we use Lemma \ref{lem:maj} with suitable
choices for $a_i$, $b_i$. Let $\ba = [a_n\ a_{n-1}\ \cdots\ a_1]$ be defined as follows:
\begin{equation}
  \label{eq:31}
\ba = [\underbrace{\cdots\ \nu\ \cdots}_{|T|\text{ times}}\ \underbrace{\cdots\ N_m\ \cdots}_{m\in\hat{T}_i}\ \underbrace{\cdots\ \nu_{ij}\ \cdots}_{|T_{ij}|\text{ times}}],  
\end{equation}
where the numbers in the set $\{N_m:m\in\hat{T}_i\}$ are arranged in non-increasing order. Similarly, the vector $\bb = [b_n\ b_{n-1}\ \cdots\ b_1]$ is defined to be
\begin{equation}
  \label{eq:32}
  \bb = [\underbrace{\cdots\ N_l\ \cdots}_{l\in\hat{T}}\ \underbrace{\cdots\ \nu_j\ \cdots}_{|T_j|\text{ times}}\ \underbrace{\cdots\ \nu_{i}\ \cdots}_{|T_i|\text{ times}}],  
\end{equation}
where, once again, the numbers in the set $\{N_l:l\in\hat{T}\}$ are
arranged in non-increasing order. Firstly, $\ba$ and
$\bb$ defined above have the same length by \eqref{eq:26} in Lemma
\ref{lem:sum}. Further, by Lemmas \ref{lem:ord} and \ref{lem:sum},
$\ba$ and $\bb$ satisfy the constraints \eqref{eq:28}, \eqref{eq:inside} of Lemma \ref{lem:maj}. Hence, by the result of Lemma \ref{lem:maj}, we have 
$$\prod^n_{j=1} b_j= \nu_i^{|T_i|}\ \nu_j^{|T_j|}\ \prod_{l\in \hat{T}}N_l\ge \prod^n_{i=1} a_i=\nu^{|T|}\ \nu_{ij}^{|T_{ij}|}\ \prod_{m\in\hat{T}_i}N_m.$$
\end{proof}

Next, we present an important problem of online basestation allocation, where we use the sub-modularity of water-filling to show that the sum-rate obtained by a greedy online algorithm is no less than $2$ times that of an optimal offline algorithm.

\section{Online Basestation Allocation}
\label{sec:basest-alloc-probl}
Consider the scenario where $n$ mobile users arrive one at a time into
a geographical area with $m$ basestations with $m<n$. User $i$ has a
non-negative weight $w_{ij}$ to basestation $j$, which is the
signal-to-noise ratio (SNR) of the user's channel to the basestation. 
The weights $w_{ij}$ are collected into an $n\times m$ matrix $W$ for
easy reference. 

In the downlink basestation allocation problem, each user is to be
allotted to exactly one of the $m$ basestations. An arbitrary
allocation is specified by the family of sets $\M=\{M_j:
1\le j\le m\}$, where $M_j$ is the set of users allotted to
basestation $j$. Note that the sets $M_j$ partitions the set of
users. The goal is to maximize a utility function over the allocations $\M$.

We will enforce two important conditions on the allocation policy: (1)
each user needs to be allotted to a
basestation immediately upon arrival, and (2) this allocation is
irrevocable and cannot be altered subsequently. 
\subsection{Offline versus online allocation}
In an {\it offline} allocation problem, the entire weight matrix $W$
is available ahead of time. So, an offline algorithm solves a
well-defined optimization of the utility over all possible
allocations. Let us denote the optimal offline utility by $R^*_\off(W)$.

In an {\it online} allocation problem, the users arrive in an
arbitrary order, and weights of user $i$ are revealed only upon
arrival. i.e. the $i$-th row of $W$ is revealed when user $i$
arrives. 
 
Online algorithms are, by definition, poorer than offline
algorithms. {\it Competitive ratio} is a popular figure of merit used
to characterize and compare online and offline algorithms
\cite{BorodinOnlineBook}. For a given weight matrix $W$, the
competitive ratio of an online algorithm $\A$ is defined as 
$\eta_W(\A) = R^*_{\off}(W)/R_{\A}(W)$,
where $R_{\A}(W)$ denotes the utility of the online algorithm $\A$ on
input $W$. Since  $R^*_{\off}(W)\ge R_{\A}(W)$, we have $\eta_W(\A)\ge 1$ for all
$W$ and $\A$. The worst-case competitive ratio of an online algorithm $\A$ is defined as
$\eta_\worst(\A) = \max_W \eta_W(\A)$. 
The design goal is to have online algorithms with a competitive ratio as close to 1 as possible. 
\subsection{Utility and greedy online algorithm}
We suppose that each basestation has unit transmit power that is allocated to each of its connected users such that it maximizes the $\log$-utility.
Thus, to
maximize the rate, basestation $j$ allocates power $\alpha_{ij}$ to all its connected users $i\in M_j$ by
performing the following maximization:
\begin{equation}
  \label{eq:10}
 L(M_j)=\max_{\alpha_{ij}, \sum_{i\in M_j}\alpha_{ij}\le 1}\sum_{i\in M_j}\log(1+\alpha_{ij}w_{ij}), 
\end{equation}
where the weight $w_{ij}$ is the channel SNR of user $i$ to
basestation $j$. We consider the online basestation allocation
problem, where each user on its arrival is assigned to one of the base stations, i.e. to one of sets $M_j, j=1, \dots, m$ without knowing any information about the future user arrivals, so as to maximize the overall utility
 \begin{equation}
  \label{eq:logutility}
  LS(\M,W)=\sum_{j=1}^mL(M_j),
\end{equation}
which is the sum rate of the entire system over all basestations. 

We propose an online algorithm for basestation allocation that
maximizes the sum-rate \eqref{eq:logutility} at a worst-case
competitive ratio of at most $2$. The algorithm and proof use results
on the multi-partitioning problem in \cite{Nemhauser1978, Lehmann2006combinatorial}.

\begin{definition}{\bf Multi-partitioning problem:}
The problem is to partition a given set $\bS$  into $k$ subsets $S_1, \dots, S_k, S_i \cap S_j = \phi, \cup_{i=1}^k S_i = \bS$ such that $\sum_{i=1}^k f_i(S_i)$ is maximized.
\end{definition}

\begin{table}
\begin{tabular}{r l}
\hline
& \textbf{ONLINE GREEDY ALLOCATION}\\
\hline
1	&	Initialize  $S_j = \Phi$, $j =1,\ldots, k$, $i=1$. \\
2	&	Find $j^*=\arg\max_{j}f_{j}(S_j\cup\{i\})$.	\\
3 &	Update $S_{j^*} = S_{j^*} \cup \{i\}$.\\
4	&	Set $i=i+1$, Stop if $i>|\bS|$, otherwise go to step 2.\\
5	&	Return $S_j$, $j =1,\dots, k$.\\
\hline
\end{tabular} 
\label{tab:greedy}
\end{table}

\begin{theorem}[\cite{Nemhauser1978,Lehmann2006combinatorial}]\label{theorem:Nemhauser}
  If all functions $f_i$ in the multi-partitioning problem are
  non-negative, monotone and sub-modular, then the ONLINE GREEDY ALLOCATION
  algorithm (presented above) has a competitive ratio of at
  most $2$.
\end{theorem}

Note that the sum-rate maximization problem \eqref{eq:logutility} can be thought of as a multi-partioning problem, 
where $f_j = L(M_j)$. To make use of 
Theorem \ref{theorem:Nemhauser}, we need to check that the function
$L(M_j)$ is sub-modular. Now, note that \eqref{eq:10} is a special
case of the water-filling function \eqref{eq:7}, where $P_i$ plays the
role of $\alpha_{ij}$, sum-power constraint $P=1$, and $w_{ij} = \frac{1}{N_i}$. Also, note
that \eqref{eq:10} is non-negative and monotone. Therefore, we have
the following:
\begin{theorem}\label{theorem:finaltheorem}
  The competitive ratio  of the ONLINE GREEDY ALLOCATION
  algorithm is $\le2$ for maximizing the sum-rate
  \eqref{eq:logutility}.
\end{theorem}
\begin{proof} Follows from Theorem \ref{theorem:Nemhauser}, and Theorem \ref{thm:submod}, where $|\bS| = n$, $S_j = M_j$, $j=1,\dots,m$, and $f_j = L(M_j)$ for the ONLINE GREEDY ALLOCATION algorithm.
\end{proof}

\section{Simulation Results}
In Fig. \ref{fig:logutility}, we plot an upper bound on the competitive ratio for the GREEDY ALLOCATION algorithm with $m = 10$ basestations 
under different user SNR profiles. To upper bound the competitive
ratio, we upper bound the utility of the offline algorithm by assuming
that each user has all SNRs equal to maximum SNR across all the users
and all the basestations. Under this assumption, allocating equal number of users to each basestation is optimal. First, we consider the i.i.d. case, where
each user's SNR is distributed uniformly in either $[0, 1]$ or $[0,
10]$ to all the basestations. We see that in both the cases the
competitive ratio upper bound is very close to $1$, which is
significantly better than the derived worst case bound of $2$.

Next, we consider two non-i.i.d. cases for the user SNRs -- 
case $1$, where all the SNRs of half of the users 
are uniformly distributed in $[0, 10]$, while the SNRs of the other half of 
the users are uniformly distributed in $[0, 5]$; and case $2$, where
for each user, SNRs to 3 randomly selected basestations are uniform in
$[0, 10]$ and other SNRs are uniform in $[0, 1]$. 
This creates users SNRs with different orders of 
magnitude. Here also the competitive ratios are very close to unity.

Finally, we consider a case of correlated user SNRs, 
where for each user 3 randomly selected basestations have 
the same SNR values while the other basestations have half that value.
We see that even for this highly correlated case, competitive ratio is 
below the stipulated value of $2$, which clearly 
validates Theorem \ref{theorem:finaltheorem}.

For comparison we also plot the competitive ratios of MAX-WEIGHT strategy, 
which assigns users to the basestation to which it has the highest 
SNR, for the i.i.d. case 1 and the correlated case. 
In the former the competitive ratios are the same for 
both the algorithms, but in the latter, GREEDY algorithm performs
significantly better than the MAX-WEIGHT strategy.

\begin{figure}
\centering
\includegraphics{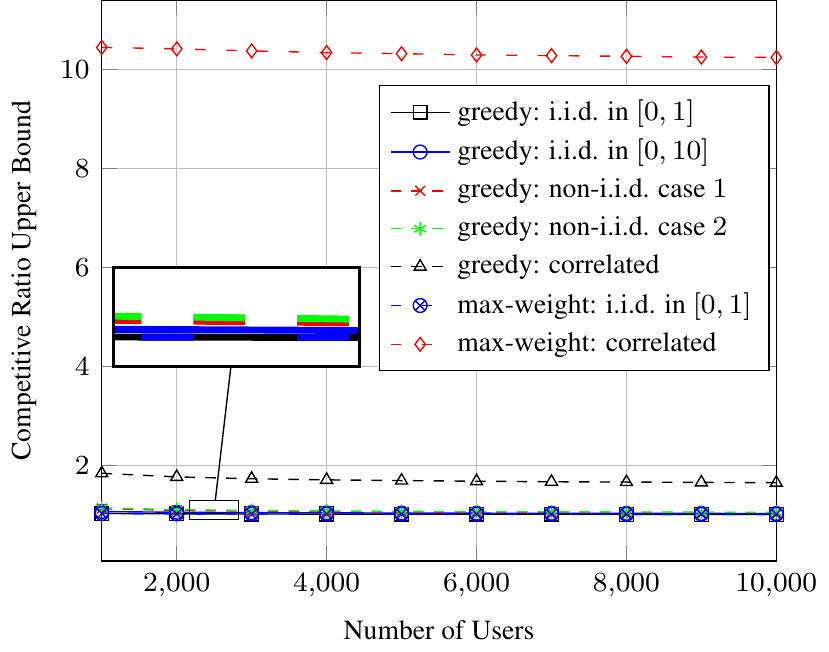}
\caption{Competitive ratio of different algorithms with $m=10$ basestations.}
\label{fig:logutility}
\end{figure} 

\appendices
\section{}
\label{sec:simpl-cond-subm}
Using (\ref{eq:9}), in order for the submodularity condition (\ref{eq:usefuldefnsubmod}) to be true, we need to show that
$$\log\frac{\nu_i^{|T_i|}}{\prod_{k\in T_i}N_k}+\log\frac{\nu_j^{|T_j|}}{\prod_{k\in T_j}N_k}  
\ge \log\frac{\nu^{|T|}}{\prod_{k\in T}N_k}+\log\frac{\nu_{ij}^{|T_{ij}|}}{\prod_{k\in T_{ij}}N_k}, $$
 which simplifies to
\begin{equation}
  \label{eq:16}
  \frac{\nu_i^{|T_i|}}{\prod_{k\in T_i}N_k}\ \frac{\nu_j^{|T_j|}}{\prod_{k\in T_j}N_k}\ge \frac{\nu^{|T|}}{\prod_{k\in T}N_k}\ \frac{\nu_{ij}^{|T_{ij}|}}{\prod_{k\in T_{ij}}N_k}.
\end{equation}
Using the decompositions in \eqref{eq:1}~-~\eqref{eq:2} for the products in the denominator of (\ref{eq:16}), we get
\begin{align*}&\frac{\nu_i^{|T_i|}}{N_i\prod_{k\in \bar{T}_{ij}}N_k\prod_{m\in \hat{T}_i}N_m}\  \frac{\nu_j^{|T_j|}}{N_j\prod_{k\in \bar{T}_j}N_k} \\ 
&\ge \frac{\nu^{|T|}}{\prod_{k\in \bar{T}_j}N_k\prod_{l\in \hat{T}}N_l}\ \frac{\nu_{ij}^{|T_{ij}|}}{N_iN_j\prod_{k\in \bar{T}_{ij}}N_k},
\end{align*}
which simplifies to \eqref{eq:17}, after canceling $N_iN_j\prod_{k\in \bar{T}_{ij}}N_k\prod_{k\in \bar{T}_j}N_k$ and cross multiplication.
\section{Proof of Lemma \ref{lem:ord}}
\label{sec:proof-lemma-ord}
We use the definition of the set $T(\cdot)$ in \eqref{eq:5}. Since $m\in\hat{T}_i=\bar{T}_i\setminus\bar{T}_{ij}$, we have that $m\in T_i$, but $m\notin T_{ij}$. So, clearly, $N_m\le\nu_i$ and $\nu_{ij}\le N_m$. This proves the inequalities $(a)$ and $(b)$. The inequality $(c)$ follows by the assumption in \eqref{eq:3}. 

Since $l\in\hat{T}=T\setminus\bar{T}_j$, we have that $l\in T$, but $l\notin T_j$. So, clearly, $N_l\le\nu$ and $\nu_j\le N_l$. This proves the final two inequalities $(d)$ and $(e)$.
\section{Proof of Lemma \ref{lem:sum}}
\label{sec:proof-lemma-sum}
  Using \eqref{eq:13}, we have 
  \begin{align}
    \label{eq:6}
    |T_i|\nu_i=P+\sum_{k\in T_i}N_k,\ \ &|T_j|\nu_j=P+\sum_{k\in T_j}N_k,\\
\label{eq:25}    |T|\nu=P+\sum_{k\in T}N_k,\ \ &|T_{ij}|\nu_{ij}=P+\sum_{k\in T_{ij}}N_k.
  \end{align}
From \eqref{eq:6} and \eqref{eq:25}, we have
\begin{equation}
  \label{eq:15}
  |T_i|\nu_i+|T_j|\nu_j+\sum_{l\in\hat{T}}N_l=2P+\sum_{k\in T_i}N_k+\sum_{k\in T_j}N_k+\sum_{l\in\hat{T}}N_l.
\end{equation}
Using the decompositions for $T_i$ and $T_j$ in \eqref{eq:1}, we have 
\begin{align}
  \label{eq:19}
  \sum_{k\in T_i}N_k&=N_i+\sum_{k\in \bar{T}_{ij}}N_k+\sum_{m\in \hat{T}_i}N_m,\\
\label{eq:20} \sum_{k\in T_j}N_k&=N_j+\sum_{k\in \bar{T}_j}N_k.
\end{align}
Using \eqref{eq:19}, \eqref{eq:20} in the RHS of \eqref{eq:15} and rearranging, we get
\begin{align}
\nonumber  |T_i|\nu_i+|T_j|\nu_j+&\sum_{l\in\hat{T}}N_l=P+\sum_{k\in \bar{T}_j}N_k+\sum_{l\in\hat{T}}N_l+\\ 
 \label{eq:21}                                                                  &P+N_i+N_j+\sum_{k\in \bar{T}_{ij}}N_k+\sum_{m\in \hat{T}_i}N_m. 
\end{align}
Using \eqref{eq:2} and \eqref{eq:25}, the RHS of \eqref{eq:21} and \eqref{eq:18} 
are seen to be equal, proving \eqref{eq:18}. To prove \eqref{eq:26}, we use \eqref{eq:1} to get
\begin{align}
\nonumber  |T_i|+|T_j|+|\hat{T}|&=1+|\bar{T}_{ij}|+|\hat{T}_i|+1+|\bar{T}_j|+|\hat{T}|,\\
\nonumber &=(|\bar{T}_j|+|\hat{T}|)+(2+|\bar{T}_{ij}|)+|\hat{T}_i|,\\
  \label{eq:30}&=|T|+|T_{ij}|+|\hat{T}_i|.
\end{align}
\section{Proof of Lemma \ref{lem:maj}}
\label{sec:proof-lemma-maj}
The first step in the proof is to show that the conditions on $a_i$,
$b_j$ imply that the vector $[a_n\ a_{n-1}\ \cdots\ a_1]$ majorizes
the vector $[b_n\ b_{n-1}\ \cdots\ b_1]$. Since the sum of the two
vectors are equal (from (\ref{eq:28})), we only need to show, for $1\le i\le n$,
\begin{equation}
  \label{eq:23}
  A_i\triangleq a_n+a_{n-1}+\cdots+a_i\ge B_i\triangleq b_n+b_{n-1}+\cdots+b_i.
\end{equation}
From (\ref{eq:inside}), it is clear that $A_i\ge B_i$ for $k+1\le i\le n$. For $1\le i\le k$, we have $a_i\le b_i$, or 
\begin{equation}
  \label{eq:24}
  A_i-A_{i+1}\le B_i-B_{i+1}.
\end{equation}
Summing \eqref{eq:24} from $i-1$ down to $i=1$, we get
\begin{equation}
  \label{eq:27}
  A_1-A_i\le B_1-B_i,\ 1\le i\le k.
\end{equation}
Since $A_1=B_1$, we have $A_i\ge B_i$, $1\le i\le k$. This proves that $[a_n\ a_{n-1}\ \cdots\ a_1]$ majorizes $[b_n\ b_{n-1}\ \cdots\ b_1]$. 

The second step is to invoke Karamata's inequality \cite{zbMATH03006493}\cite{kadelburg2005inequalities}, which states that 
\begin{equation}
g(a_1)+g(a_2)+\cdots+g(a_n) \ge g(b_1)+g(b_2)+\cdots+g(b_n),
\label{eq:29}
\end{equation}
for any convex function $g$, if $[a_n\ a_{n-1}\ \cdots\ a_1]$
majorizes $[b_n\ b_{n-1}\ \cdots\ b_1]$. Using (\ref{eq:29}) with
$g(\cdot)=-\log(\cdot)$, the claim of the Lemma in (\ref{eq:prod}) follows.

\bibliographystyle{IEEEtran}
\bibliography{IEEEabrv,Research}

\end{document}